\documentclass[12pt]{article}

\usepackage{microtype}

\usepackage{graphicx}
\usepackage{cite}
\usepackage{color}
\usepackage{amsfonts}
\usepackage{amssymb,amsmath,amsthm}
\usepackage{fullpage}

\newtheorem{lemma}{Lemma}
\newtheorem{theorem}[lemma]{Theorem}

\newtheorem{cor}[lemma]{Corollary}

\DeclareMathOperator{\Vor}{Vor}
\DeclareMathOperator{\MST}{MST}
\DeclareMathOperator{\GG}{GG}
\DeclareMathOperator{\DG}{DG}
\DeclareMathOperator{\DT}{DT}
\DeclareMathOperator{\CH}{CH}

\newcommand{\Pabs}[1]{\left|#1\right|}

\let\quasi\prec

\newtheorem{proposition}{Proposition}

\author{Boris Aronov\thanks{%
    Department of Computer Science and
    Engineering, Polytechnic Institute of NYU, Brooklyn, New
    York~~11201 USA.  Research partially supported by a grant from the
    U.S.-Israel Binational Science Foundation, by NSA MSP Grant
    H98230-06-1-0016, and NSF Grant CCF-08-30691.}
  \and 
  Muriel Dulieu\footnotemark[2]
  \and
  Ferran Hurtado\thanks{%
    Departament de Matem\`{a}tica Aplicada~II, 
    Universitat Polit\`{e}cnica de Catalunya, 
    Barcelona, Spain.
    Partially supported by projects MEC MTM2006-01267 and DURSI
    2005SGR00692.}}

\title{Witness Gabriel Graphs\thanks{Part of this research was done
    while the third author was visiting Polytechnic, supported by
    Grant AGAUR-Generalitat de Catalunya 2007 BE-1 00033.}}

\begin{document}
\bibliographystyle{plain}
\maketitle

\begin{abstract}
We consider a generalization of the Gabriel graph, the
\textit{witness Gabriel graph}.
Given a set of vertices $P$ and a set of witness points $W$ in the
plane, there is an edge $ab$ between two 
points of $P$ in the witness Gabriel graph \emph{$\GG^-(P,W)$} if and only if the closed disk with diameter $ab$ does 
not contain any witness point (besides possibly $a$ and/or $b$).
We study several properties of the witness Gabriel graph, both as a
proximity graph and as a new tool in graph drawing.
\end{abstract}

\section{Introduction}

Originally defined to capture some concept of neighborliness,
\emph{proximity graphs} \cite{toussaint91some, jaromczykrelative,
DBLP:conf/gd/BattistaLL94} can be intuitively defined as follows:
given a set $P$ of points in the plane, the vertices of the graph,
there is an edge between a pair of vertices $p,q\in P$ if some
specified region in which they interact contains no point from
$P$, besides possibly $p$ and $q$.

Proximity graphs have proved to be a very useful tool in shape
analysis and in data mining \cite{jaromczykrelative,Tou05}. In
graph drawing, a problem that has been attracting
substantial research is to explore which classes of graphs admit a
proximity drawing, for some notion of proximity, and when it is
possible to efficiently decide, for a given graph, whether such a
drawing exists \cite{DBLP:conf/gd/BattistaLL94,gdbook}. For
all these reasons, several variations and extensions have been
considered, from higher-order proximity graphs to the so-called
weak proximity drawings \cite{jaromczykrelative,DLW06}.



\begin{figure}
  \centering
  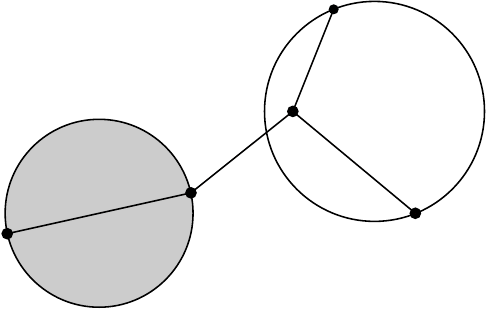%
  \hspace{1ex plus 1fil}%
  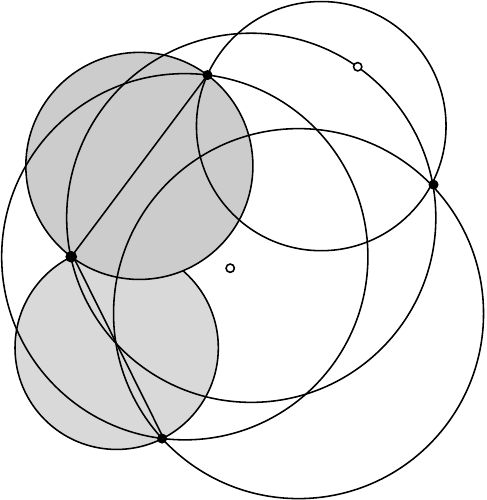
  \caption{Left:~Gabriel graph. The vertices defining the shaded disk
    are adjacent because their disk doesn't contain any other vertex,
    in contrast to the vertices defining the unshaded disk.
    Right:~witness Gabriel graph. Black points are the vertices of the
    graph, white points are the witnesses.  Each pair of vertices
    defining a shaded disk are adjacent and the pairs defining the
    remaining (unshaded) disks are not.}
  \label{GG+GGG}
\end{figure}

In the case of the \textit{Gabriel graph}, $\GG(P)$, the region of
influence of a pair of vertices $a,b$ is the closed disk with diameter
$ab$, $D_{ab}$. An edge $ab$ is in the Gabriel graph of a point set
$P$ if and only if $P\cap D_{ab} = \{a,b\}$ (see
Figure~\ref{GG+GGG}(left)).
Gabriel graphs were introduced by Gabriel and
Sokal \cite{GG} in the context of geographic variation analysis.

We consider in this work a generalization of the Gabriel graph,
the \textit{witness Gabriel graph}, $\GG^{-}(P,W)$. It is defined by
two sets of points $P$ and $W$; $P$ is the set of vertices
of the graph and $W$ is the set of
\emph{witnesses}. There is an edge $ab$ in $\GG^{-}(P,W)$ if, and
only if, there is no point of $W$ in $D_{ab} \setminus \{a,b\}$
(see Figure~\ref{GG+GGG}(right)).

Notice that witness Gabriel graphs are a proper generalization of
Gabriel graphs, because when the set $W$ of witnesses coincides with
the set $P$ of vertices, we clearly obtain $\GG^{-}(P,P)=\GG(P)$. This
was the main underlying idea of the generic concept of \emph{witness
  graphs}, which were introduced as a general framework in \cite{ADH}
to provide a generalization of proximity graphs, allowing witnesses to
play a negative role as in this paper or a positive one as
well. Several examples were described in \cite{ADH,rectangle-graphs},
including in particular witness versions of Delaunay graphs and
rectangle-of-influence drawings. A systematic study is developed in
\cite{thesis}.
As already mentioned in this introduction, generalizing basic proximity 
graphs has attracted several research efforts. This is also our main motivation. On the other 
hand, a witness graph $W(P,Q)$ is an instrument for describing the position of $P$ with respect to $Q$. 
We believe that once these graphs are well understood, by considering simultaneously $W(P,Q)$ and 
$W(Q,P)$ we would have useful tools for the description of the interaction/discrimination between the 
two sets; this is a main topic of our ongoing research.

In this paper we prove several fundamental properties of witness
Gabriel graphs, describe algorithms for their computation, and 
present results on the realizability of some combinatorial
graphs.

We assume throughout the paper that the points in $P \cup W$ are in
general position, that is, that there are no three points of $P \cup
W$ on a line and no four on a circle.

\section{Some Properties of Witness Gabriel Graphs}

It is known that $\MST(P) \subseteq \GG(P) \subseteq \DT(P)$
\cite{bb21378}, where $\MST(P)$ is the minimum spanning tree and
$\DT(P)$ is the Delaunay triangulation.  As a consequence, 
$\Pabs{\MST(P)} \leq \Pabs{\GG(P)} \leq \Pabs{\DT(P)}$, where we have used $\Pabs{\cdot}$
to denote the number of edges in a graph.  Expressing this in terms of $n = \Pabs{P}$, we have that $n-1 \leq \Pabs{\GG(P)} \leq 3n - 6$.
In \cite{MS84}, a more detailed  analysis gives a tighter upper bound of $3n - 8$.

For witness Gabriel graphs $\GG^{-}(P,W)$, the situation is quite
different, as for any fixed set $P$ of $n$ points, by varying the
size of $W$ and the location of the witnesses, the number of edges
in $\GG^{-}(P,W)$ can attain any value from 0 to $n \choose 2$.
For example, when $W=\varnothing$, we obviously obtain
$\GG^{-}(P,\varnothing)=K_n$. 

\begin{theorem}
  For any set $P$ of $n$ points in the plane, a witness Gabriel graph
  $\GG^{-}(P,W)$ can have any number of edges from 0 to $n \choose 2$
  edges, by a suitable choice of the set $W$ of witnesses.
\end{theorem}

\begin{proof}
  Consider any given graph $GG^-(P,W)$ and take the union $U$ of the
  diametral disks $D_{p_i p_j}$, $p_i, p_j \in P$, that do not contain
  a point $q \in W$.  The boundary of the union consists of circular
  arcs $C_{p_i p_j}$ of disks $D_{p_i p_j}$, for some $p_i$, $p_j$
  $\in P$.  Put a point $q \in W$ in the relative interior of one such
  arc $C_{p_i p_j} \setminus \{p_i, p_j\}$.  Point $q$ lies in the
  closed disk $D_{p_i p_j} \setminus \{p_i, p_j\}$. By construction,
  it lies outside all other disks.  Therefore adding $q$ to $W$ would
  eliminate precisely one edge, namely $(p_i,p_j)$.\footnote{%
    This choice of $q$ is in some sense degenerate, but $q$ can be
    moved slightly into the interior of $D_{p_i p_j}$ without
    affecting the argument.} %
  By iterating this procedure to remove the edges one by one from the
  witness Gabriel graph, one can see that any number of edges can be
  attained (see Figure~\ref{GGGedges}).
  \begin{figure}
    \centering
    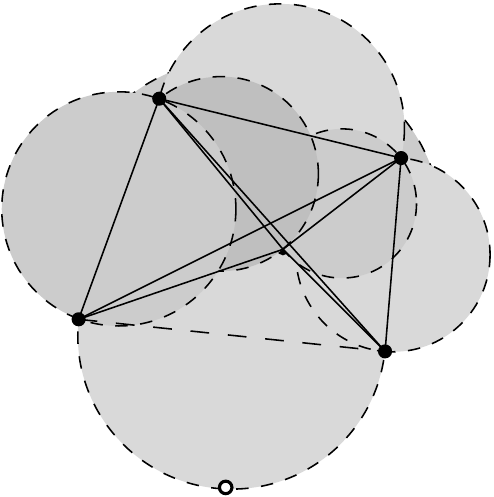
    \caption{The white point is a witness that removes the dashed
      edge.}
    \label{GGGedges}
  \end{figure}
\end{proof}

The reverse problem is more interesting: as the witness points can be
thought as \emph{interferences} that prevent the points in $P$ from
being adjacent, one may wonder how many witnesses are required to
completely eliminate all edges in $\GG^{-}(P,W)$. Trivially, if there
is a witness inside each disk $D_{ab}$, for all $a,b \in P$, then
$\GG^{-}(P,W)$ has no edge.  This can be achieved, for instance, by
putting a witness close to the midpoint of every pair $a,b$ of points
from $P$, which would give $\Pabs{W} \leq {n \choose 2}$.  In the following
theorem we present a much better bound for the number of witnesses
necessary to eliminate all edges of $\GG^{-}(P,W)$.

\begin{theorem}
\label{thm:eliminate-all} $n-1$ witnesses are always sufficient to
eliminate all edges of an $n$-vertex witness Gabriel graph, while
$\frac{3}{4}n - o(n)$ witnesses are sometimes necessary.
\end{theorem}

\begin{proof}
We start with a lower bound construction.
Place the points of $P$ at the vertices of a hexagonal tiling (see Figure~\ref{GGlower+upperbound}) 
\begin{figure}
  \centering
  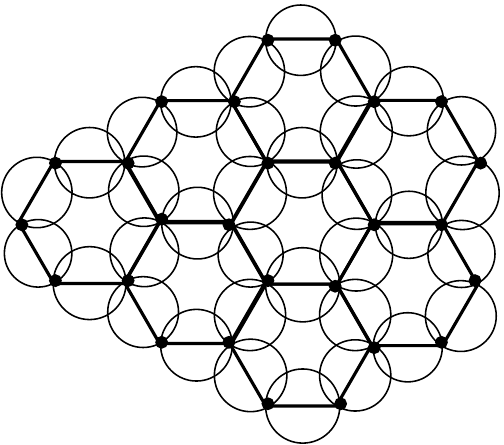%
  \hspace{0pt plus 1fil}%
  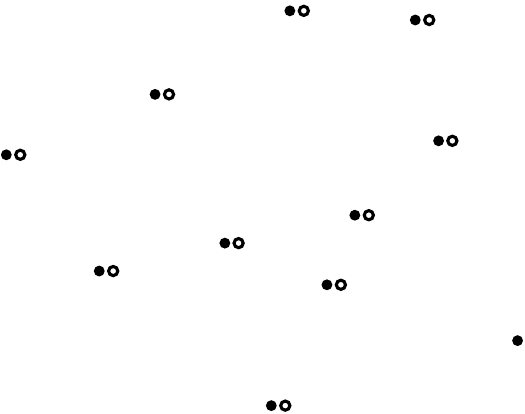
  \caption{Left: Lower bound---any point in $W \setminus P$ can
    intersect at most 2 disks.  Right:~Upper bound---points in $P$ are
    black and witnesses are white.}
  \label{GGlower+upperbound}
\end{figure}
and move them slightly so they are in general position. 
For each point $p$ in $P$ consider the three edges connecting it to the three closest points, and the closed disks that have these edges as diameter.
By definition of the witness Gabriel graph, no point $q \in W$ at the intersection of three disks eliminates one of the considered edges as $q$ would be at a point $p \in P$ defining the three disks.
For any other position of $q \in W$, $q$ never intersects more than 2 disks.
Hence since we have $\frac{3}{2} n - o(n)$ disks, at least $\frac{3}{4} n - o(n)$ points in $W$ are necessary to stab all the disks and eliminate the corresponding edges in the witness Gabriel graph.


Now we argue the upper bound.  Without loss of generality, assume no
two points of $P$ lie on the same vertical lines---this can be
achieved by an appropriate rotation of the coordinate system.  Put a
witness slightly to the right of each point of $P$, except for the
rightmost one (see Figure~\ref{GGlower+upperbound}).
Every disk with diameter determined by two points of $P$ will contain a witness.
\end{proof}

According to the preceding result, $n-1$ suitable witnesses can
always eliminate \emph{all} the edges. Interestingly, realizing some
witness Gabriel graphs that are not empty of edges may require strictly more witnesses.

\begin{theorem}
\label{thm:eliminate-some} For arbitrarily large $n$, there exist witness Gabriel graphs on
$n$ vertices that are not empty of edges, for which at least $\frac{3}{2} n - 8$ witnesses are
necessary.
\end{theorem}

\begin{proof}

Put the vertices on concentric circles, sixteen vertices evenly spaced per circle, the first vertex on each circle always on top of the disk (see Figure~\ref{GGNumWUpB}).
The ratio between two consecutive circle radii must be between $1{:}1.82$ and $1{:}1.92$.
On each circle, number the vertices clockwise, starting at the top.
Number circles starting at the innermost.
The edges of the geometric graph are the ones between consecutive even vertices on even circles and the ones between consecutive odd vertices on odd circles (see Figure~\ref{GGNumWUpB}).
\begin{figure}
  \centering
  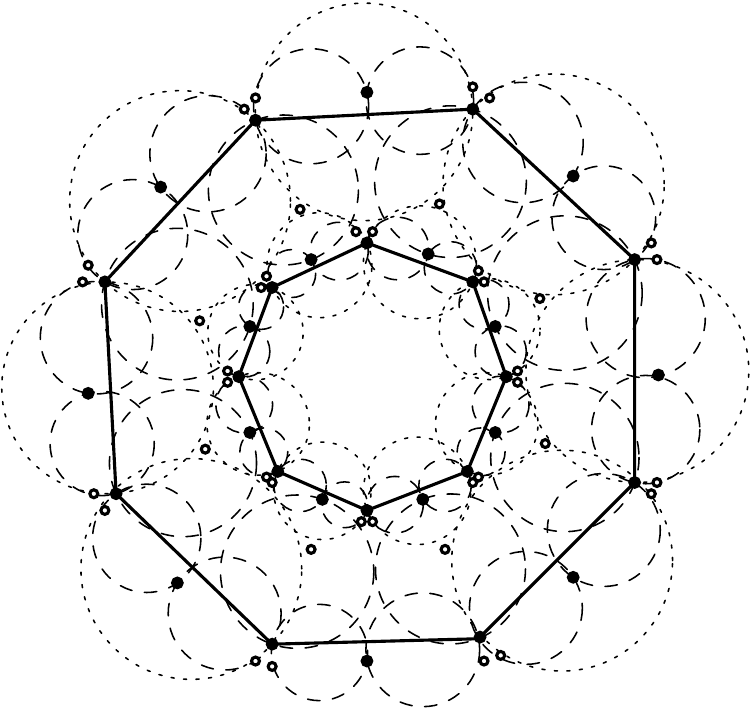
  \caption{At least $\frac{3}{2} n - 8$ witnesses (the white points)
    are sometimes necessary. The black points are the $n$ vertices of
    the graph.}
  \label{GGNumWUpB}
\end{figure}

Exactly $n$ witnesses are necessary to remove the edges between every
pair of consecutive vertices on each circle, without removing the
edges between even (or odd) vertices.  In addition, $\frac{n}{2} - 8$
witnesses are necessary to remove the edges between corresponding
vertices on two consecutive circles. Summing up, we see that
$\frac{3}{2} n - 8$ witnesses are required, as claimed.
\end{proof}

\section{Witness Gabriel Drawings}

Given a combinatorial graph $G=(V,E)$, there is a \textit{witness Gabriel
drawing} of $G$ if there is a point set $P$ and a set of witnesses $W$ such that the witness Gabriel graph of $P$ and $W$ is isomorphic to $G$. 
This witness Gabriel graph of $P$ and $W$ is a witness Gabriel drawing of $G$.

The fundamental question concerning witness Gabriel
drawability is the following: Given a graph $G$, is it possible to
construct a witness Gabriel drawing of $G$?
That question has been studied for (standard) Gabriel graphs. In
\cite{bose93characterizing} Bose et al.~present a complete
characterization of those trees that are drawable as a Gabriel
graph. They proved that such trees cannot have vertices of degree greater
than four and cannot have two adjacent vertices of degree four.
They also characterized Gabriel-drawable trees by exhibiting families
of forbidden subtrees and by showing that they don't contain members of
these families.

Lubiw and Sleumer \cite{LS93} showed that all maximal outerplanar
graphs admit a Gabriel drawing in the plane. They also conjectured
that any biconnected outerplanar graph has a Gabriel drawing. This
was settled in the affirmative by Lenhart and Liotta \cite{
lenhart96proximity}.

As every Gabriel graph is also a witness Gabriel graph, since
$\GG(P)=\GG^{-}(P,P)$, one may expect the witness Gabriel graphs to be
a more powerful tool for representing graphs, compared to classical
Gabriel graphs.  This is indeed the case for trees.

\begin{theorem}
  \label{thm:trees}
  Any tree can be drawn as a witness Gabriel graph.
\end{theorem}

\begin{proof}

  We construct a witness Gabriel drawing of a given tree $T$ as
  follows.  We assume, without loss of generality, that the tree is
  rooted.  Draw the root of $T$ as an arbitrary point.  Number the
  nodes of $T$ arbitrarily with 1 as root.  Order the children of
  every node arbitrarily.

  With each node $j$, we associate two values : $\alpha_{j}$ is the
  measure of the angles $\angle kjl$ with $k,l$ being two consecutive
  children of $j$, and $d_j$ is the number of children of the node
  $j$.  Whenever it causes no confusion, in the remainder of the proof
  we do not distinguish between a vertex of $T$ and the point
  representing it.

  For the special case of the root (node~1), $\alpha_{1}$ is set to
  $360^\circ / (d_1+1)$.  For every other node $j$ we define
  $\alpha_{j} = \alpha_{h} / d_j$ , with $h$ being the parent of $j$.

  Recursively, for each node $j$, beginning with the root $1$, draw
  the $d_j$ children such that the angles between two edges incident
  to two consecutive children of $j$ are $\alpha_{j}$ and the angles
  between the edges incident to an extremal child and the parent of
  $j$ are $\frac{360^\circ - \alpha_{j} \times (d_j -1)}{2}$.

  The length of the edges is defined as follows.  All the edges
  incident to the root have length 1.  Consider a node $j$ at depth
  $i$, its parent $h$ and its child $k$.  We set the length $\Pabs{jk}$ of
  the edge $jk$ to $\frac{1}{2} \Delta_j$, with $\Delta_j = \Pabs{hj} \sin
  \frac{\alpha_{h}}{2}$ (see Figure~\ref{GGTree1}).  
  \begin{figure}
    \centering
    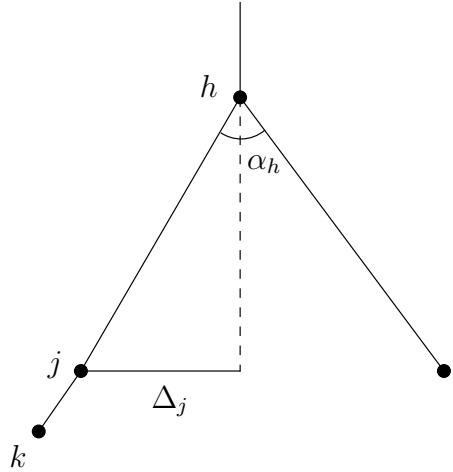
    \caption{vertices $h,j,k$ and the value of $\Delta_j$.}
    \label{GGTree1}
  \end{figure}

  The way the angles between consecutive children and the length of
  the edges are defined ensure that two edges never cross.  Indeed, by
  construction, the Euclidean length of any path connecting $j$ to its
  descendant in the tree is shorter than $\Delta_j$ and the same is
  true of any of its siblings $j'$.  Since disks of radius
  $\Delta_j=\Delta_{j'}$ centered at $j$ and $j'$ do not overlap, the
  two paths cannot cross.

  Now we shall place the witnesses.  For every edge $jh$ of the new
  tree connecting a non-root node $j$ to its parent $h$, draw a
  Gabriel disk $D_{jh}$ with diameter $jh$.  Place two witnesses
  $w_1$ and $w_2$ on both sides of $j$ such that they make an angle
  $\measuredangle jhw_i = \frac{\alpha_h}{2} $, $i=\{1,2\}$ and such that they
  are very close but outside the Gabriel disk $D_{jh}$ (and outside
  the Gabriel disks of $h$ and its other children) (see
  Figure~\ref{GGTree2}).
  \begin{figure}
    \centering
    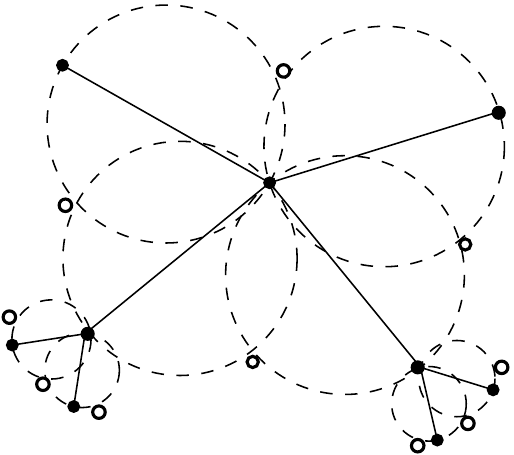
    \caption{The vertices of the graph are black and the witnesses are
      white.}
    \label{GGTree2}
  \end{figure}

  Let $F$ be the intersection of two half planes $H_1$, $H_2$ defined
  as follows: $H_i$ is the half plane containing $j$ bounded by a line
  through $w_i$ and perpendicular to $j w_i$.  By construction $j$ and
  all of its descendants lie in $F$ since they are contained in a disk
  of radius strictly smaller than $\Delta_j$ centered at $j$.  The
  placement of $w_1$, $w_2$ guarantees that no edge exists between 
  any vertex in the subtree of $j$ and
  the vertices outside of $F$ (see Figure~\ref{GGTreeWedge}).
  \begin{figure}
    \centering
    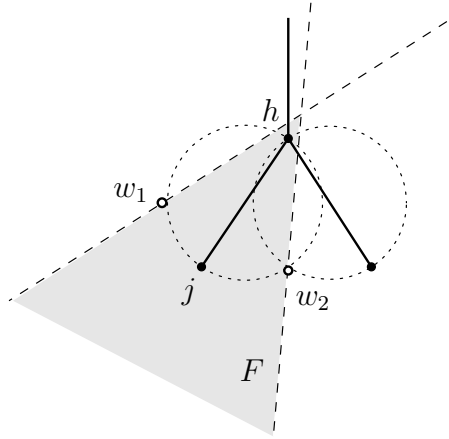
    \caption{The wedge $F$ around $j$ ensures that $j$ will be
      connected only to its parent and its children.}
    \label{GGTreeWedge}
  \end{figure}
  Any point $p \in P$ outside $F$ would make one of the angles $\angle
  p w_i j$, $i =\{1,2\}$, larger than $90^\circ$ and therefore would
  not be connected to $j$.  Applying this reasoning to every node
  $j\neq 1$ in $T$, we conclude that each node is connected only to
  its parent and its children.
  Therefore $GG^-(P,W)$ is isomorphic to $T$.
\end{proof}

We also prove that one can draw any complete bipartite graph as a
witness Gabriel graph.

\begin{theorem}
  Every complete bipartite graph can be drawn as a witness Gabriel graph.
\end{theorem}
\begin{proof}
  We construct a drawing of $K_{m,n}$ with $m \geq n$.  To avoid
  trivialities, we assume that $m>1$.  Draw an axis-aligned rectangle
  $acdb$ such that the diametral disk $D=D_{ad}=D_{cb}$ is a unit
  disk, i.e., $\Pabs{ad}=\Pabs{cb}=1$.  Let $p$ be any point on the segment $ab$
  and let $q$ be any point on the segment $cd$ (see
  Figure~\ref{GGKmn}(left)).  Let $S$ be the horizontal strip bounded
  by the lines $ab$ and $cd$.  The diametral disk $D_{pq}$ (dashed in
  the figure) is interior to $S$ but for two circular segments that are
  contained in the circular segments determined by the chord $ab$ and
  the chord $cd$ on $D$, and its center lies on the line parallel to
  $ab$ and $cd$ through the center of the rectangle. Therefore, if we
  place no witness in $D\cup S$, then $p$ and $q$ would 
  necessarily be adjacent in a witness Gabriel graph that 
  included $p$ and $q$ as vertices.

  Now, put $m$ points $a=p_1,p_2,...,p_m=b$ equally spaced on $ab$,
  and $n$ points $c=q_1,p_2,...,q_n=d$ equally spaced on $cd$, and let
  $P$ be the set of these $m+n$ points.

Consider now a disk $D'$ having as diameter two consecutive points
on a horizontal edge of the rectangle, say $p_i$ and $p_{i+1}$ on
segment $ab$. If the radius of $D'$ is bigger than the height of
the circular segment defined by the chord $ab$ on $D$, then $D'$
will stick out of the disk $D$ (see Figure~\ref{GGKmn}(right)).

The radius of $D'$ is $\Pabs{ab}/2(m-1)$, hence the preceding condition
translates to  $\Pabs{ab}/(m-1)> \Pabs{cb} - \Pabs{ac}$. Equivalently, if we
denote $x=\Pabs{ab}$,  taking into account that $\Pabs{cb}=1$, the condition
becomes
\begin{equation}
\frac{1}{m-1}>\frac{1 -\sqrt{1-x^2}}{x}. \label{eqno1}
\end{equation}
Since the right-hand side of Eq.~\eqref{eqno1} tends to zero as $x$
tends to zero, no matter how large $m$ is, we can always select a
value of $x$ such that condition~\eqref{eqno1} is satisfied. Therefore
taking this value of $x$ all the disks having as diametral pair two of
the $p_i$'s (or two of the $q_j$'s), whether consecutive or not, will
stick out of the region $D\cup S$ and can be pierced by a set of
witnesses $W$, all of them outside the region $D\cup S$.  
Therefore the witness Gabriel graph $GG^-(P,W)$ is isomorphic to $K_{m,n}$, and
the claim is proved.
Note that an infinitesimal perturbation of the points in $P\cup W$ would remove
collinearities and still $GG^-(P,W)\simeq K_{m,n}$.
\end{proof}

\begin{figure}[htbp!]
  \centering
  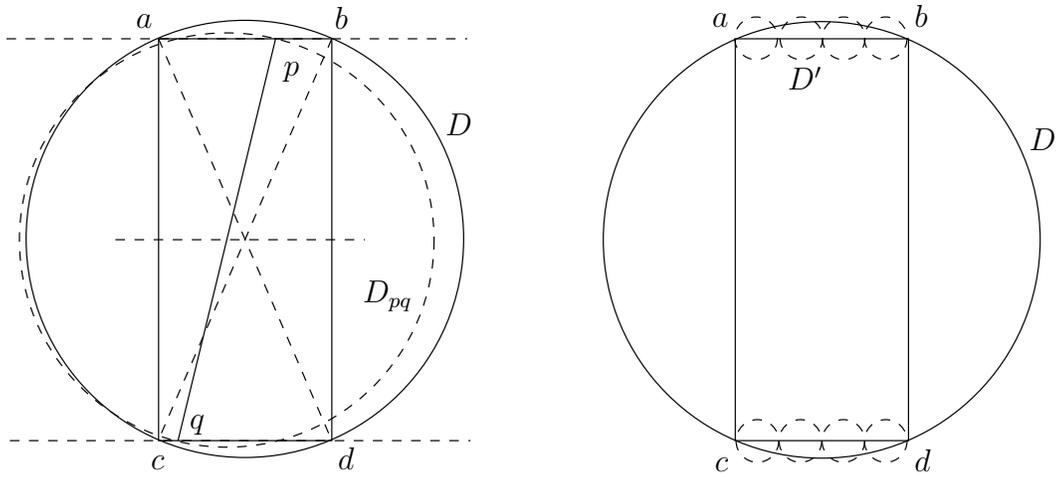
  \caption{Left: Making sure that no witness eliminates any edge $pq$
    for $p \in ab$ and $q \in cd$.  Right:~Ensuring that small disks
    $D'$ are not completely covered by the large disk $D$.}
  \label{GGKmn}
\end{figure}

In the following, we observe some properties of witness Gabriel
drawings, before deducing that some complete $k$-partite graphs, $k >
2$, have no such drawings.

\begin{lemma}
  \label{wedge}
  In a witness Gabriel graph $GG^-(P,W)$, for any pair of incident
  edges $ab$ and $bc$, all the points $p \in P$ in the triangle
  $\triangle abc$ are connected to the common vertex $b$, and there is
  no witness in $\triangle abc$ (except possibly for the vertices $a,
  b, c$, if they belong to $P\cap W$).
\end{lemma}

\begin{proof}
Let $d$ be the foot of the perpendicular from $b$ to $ac$ (see Figure~\ref{GGEdgesIncidents1+2}).
\begin{figure}
  \centering
  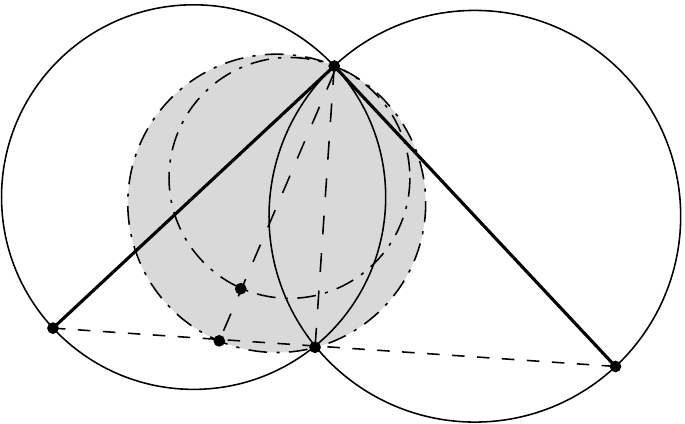%
  \hspace{0pt plus 1fil}%
  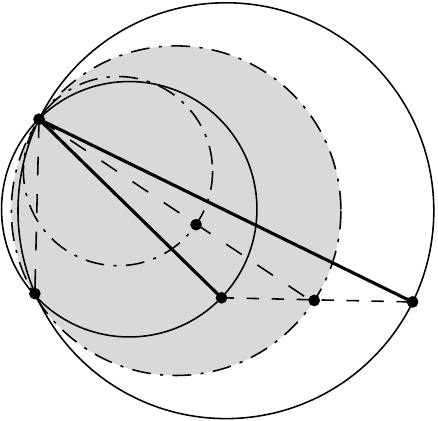
  \caption{$D_{be}$ is included in $ D_{ab} \cup D_{bc}$.}
  \label{GGEdgesIncidents1+2}
\end{figure}
This point is at the intersection of the boundaries of the disks
$D_{ab}$ and $D_{bc}$ of diameters $ab$ and $bc$, respectively.  For
any edge $be$ with $e \in ac$, $\measuredangle bde = \measuredangle bda =
\measuredangle bdc = 90^\circ$.  For any $e \in ac$, $D_{be} \subset D_{ab}
\cup D_{bc}$.  As $D_{ab}$ and $D_{bc}$ don't contain any witnesses,
neither does $D_{be}$.

Consider now a point $f$ inside the triangle $\triangle abc$.
Extend the segment $bf$ until it meets $ac$ at a point $e$.
Now $D_{bf} \subseteq D_{be} \subseteq D_{ba} \cup D_{bc}$.
Therefore $D_{bf}$ is empty of witnesses and $bf$ is an edge of $GG^-(P,W)$.
This proves the first part of the Lemma.

For the second part, we use the fact that the two disks $D_{ab}$ and $D_{bc}$ with diameters $ab$ and $bc$ cover the interior of the triangle $\triangle abc$; therefore, there is no witness in $\triangle abc$ as any such witness inside the triangle would remove one or both edges $ab$ and $cd$.
\end{proof}

\begin{proposition}
If a witness Gabriel graph has as a subgraph a triangle $\triangle a b c$, the vertices $a$, $b$, $c$ and the vertices inside this triangle form a complete subgraph (see Figure~\ref{GGTriangle}).
\begin{figure}
  \centering
  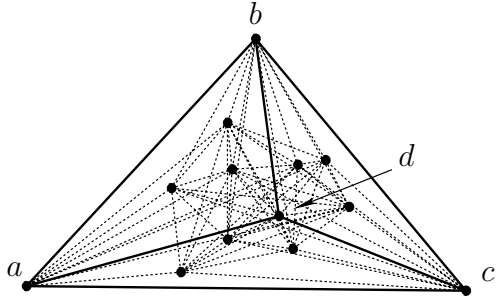
  \caption{All the vertices inside the triangle $\triangle abc$ will
    be connected to $d$.}
  \label{GGTriangle}
\end{figure}
\end{proposition}

\begin{proof}
Considering the three pairs of edges $ab$ and $bc$, $bc$ and $ca$, and $ca$ and $ab$, by Lemma~\ref{wedge}, all the vertices inside the triangle are connected to $a$, $b$, and $c$.

To complete the proof, consider a vertex $d$ inside the triangle (see
Figure~\ref{GGTriangle}).  We have already seen that it will be
connected to the vertices $a$, $b$ and $c$.  The edges $ad$, $bd$ and
$cd$ together with the triangle edges define three new triangles
$\triangle abd$, $\triangle acd$ and $\triangle bcd$.  All the other
points inside the triangle $\triangle abc$ are inside one of these
three triangles.  Therefore they will be connected to $d$, by another
application of Lemma~\ref{wedge}.
\end{proof}

We denote by $K_{i,i,i}$ the $i \times i \times i$ complete 3-partite
graph; we associate a color to each part; similarly, $K_{i,i,i,i}$
denotes the $i \times i \times i \times i$ complete 4-partite graph.
In the following, we consider only complete $k$-partite graphs with $k
\geq 3$ and at least two vertices of each color, as we have seen that
$K_{m,n}$ can be drawn as a witness Gabriel graph for any $m$ and $n$.


\begin{lemma}
  \label{2 vertices}
  In every witness Gabriel drawing of a complete $k$-partite graph,
  there are at least two vertices of each color on the convex hull $H$
  of the set of vertices.
\end{lemma}
\begin{proof}
  First suppose that for some color, say black, there are no black
  vertices on $H$.  Therefore there are at least two black vertices
  $b_1$, $b_2$ inside $H$.  Draw all the edges between $b_1$ and the
  vertices of the hull (see Figure~\ref{GG1vOnCH}(left)).
  \begin{figure}
    \centering
    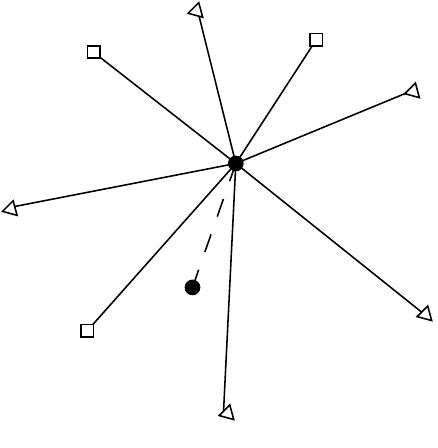%
    \hspace{0pt plus 1fil}%
    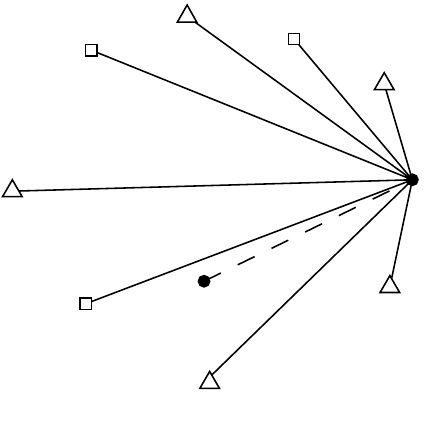
    \caption{Left: No black vertex on $H$. Right:~One black vertex on
      $H$.}
    \label{GG1vOnCH}
  \end{figure}
  These are edges of the graph, so they must be in the witness Gabriel
  drawing.  As $b_2$ is inside $H$, it is in a triangle $\triangle b_1
  m g$ defined by two edges incident to $b_1$.  By Lemma~\ref{wedge},
  $b_2$ is adjacent to $b_1$, a contradiction.

  Now suppose there is exactly one black vertex $b_1$ on $H$.
  Therefore there is at least one other black vertex $b_2$ inside $H$.
  Draw all edges between $b_1$ and other vertices of $H$ (see
  Figure~\ref{GG1vOnCH}(right)).  These edges have to be in the witness
  Gabriel drawing.  As $b_2$ is inside $H$, it is in the triangle
  defined by two edges incident to $b_1$.  By Lemma~\ref{wedge},
  $b_2$ is adjacent to $b_1$, a contradiction.
\end{proof}

\begin{lemma}
\label{GGK222} It is impossible, in a witness Gabriel drawing of
a complete $k$-partite graph, for a line containing two vertices
of one color, to divide the plane in two half-planes, each of them
containing two points of a second and third color respectively
(see Figure~\ref{fig:GGK222}(left)).  The second and the third colors may be
the same.
\begin{figure}
  \centering
  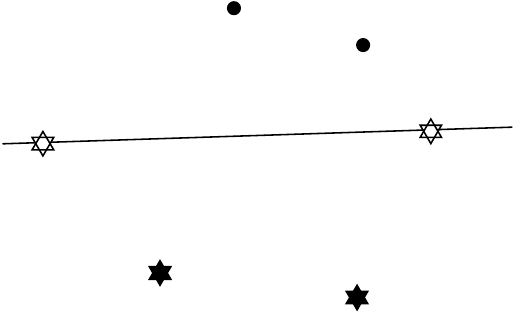%
  \hspace{0pt plus 1fil}%
  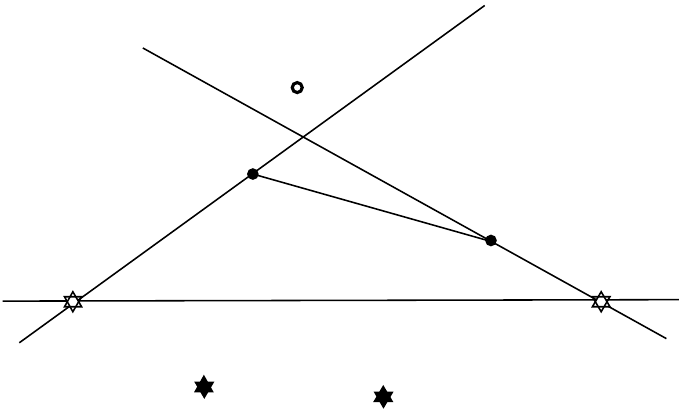
  \caption{Left: The line containing the white-star points separates two
    black-dot points from two black-star points.
    Right:~Illustration to the proof of Lemma~\ref{GGK222}.}
  \label{fig:GGK222}
\end{figure}
\end{lemma}

\begin{proof}
  Consider two white-star vertices $v_{b1}$, $v_{b2}$, two black-point
  vertices $v_{m1}$, $v_{m2}$, and two black-star vertices $v_{g1}$,
  $v_{g2}$.  We will show that it is not possible to place the
  witness(es) to remove the white-star monochromatic edge, the
  black-point monochromatic edge, and the black-star monochromatic
  edge without removing a bi-chromatic edge.  Let the witnesses
  removing the edges $v_{g1} v_{g2}$, $v_{b1} v_{b2}$, and $v_{m1}
  v_{m2}$ be $w_g$, $w_b$, and $w_m$, respectively.  They need not be
  distinct.

  Without loss of generality, assume $w_b$ lies above the line $v_{b1}
  v_{b2}$ (see Figure~\ref{fig:GGK222}(right)), i.e., on the side
  containing $v_{m1}$ and $v_{m2}$.  Note that the vertices $v_{b1}$,
  $v_{b2}$, $v_{m1}$, $v_{m2}$ must be in convex position, as
  otherwise their convex hull would be a triangle, for example
  $v_{b1}v_{b2}v_{m1}$ with $v_{m2}$ in its interior, implying by
  Lemma~\ref{wedge} $v_{m1}$ and $v_{m2}$ be adjacent.  Therefore, the
  four vertices $v_{b1}$, $v_{m1}$, $v_{m2}$, and $v_{b2}$ are in
  convex position and, without loss of generality, in this order
  around their convex hull.  Moreover, applying Lemma~\ref{wedge} to
  the bi-chromatic edges $v_{b1} v_{m1}$ and $v_{b1} v_{m2}$ (and to
  $v_{m2} v_{b1}$, $v_{m2} v_{b2}$), we conclude that there is no
  witness in the convex hull of the four vertices.

  The witness $w_b$ cannot lie to the left of the line $v_{b2} v_{m2}$
  as it would remove the edge $v_{b1} v_{m2}$.  Indeed, as
  $\measuredangle v_{b1} w_b v_{b2} \geq 90^\circ$, for $w_b$ to the
  left of $v_{b1}v_{m2}$, $\measuredangle v_{b1} w_b v_{m2} >
  \measuredangle v_{b1} w_b v_{b2} \geq 90^\circ$, and the edge
  $v_{b1} v_{m2}$ would be removed.  Symmetrically, $w_b$ cannot lie
  to the right of the line $v_{b1} v_{m1}$ as it would remove the edge
  $v_{b2} v_{m1}$.  Therefore, the lines $v_{b1} v_{m1}$ and $v_{b2}
  v_{m2}$ must cross and $w_b$ lies in region $A$ to the right of the
  line $v_{b2} v_{m2}$ and to the left of the line $v_{b1} v_{m1}$
  (see Figure~\ref{fig:GGK222}(right)).

  Now consider the four regions $A$, $B$, $C$, $D$, external to the
  quadrilateral $v_{b1}v_{b2}v_{m2}v_{m1}$, as in
  Figure~\ref{fig:GGK222}(right).  Witness $w_m$ cannot be in region
  $C$ or $D$---irrespective of whether it lies above or below the line
  $v_{m1}v_{m2}$, its presence in $C$ or $D$ would eliminate at least
  one of the edges $v_{b1} v_{m2}$, $v_{b2} v_{m1}$. Indeed,
  $\measuredangle v_{m1} w_m v_{m2} < \max\{ \measuredangle v_{b1} w_m
  v_{m2}, \measuredangle v_{m1} w_m v_{b2} \}$ for all placements of
  $w_m$ in $C$ or $D$.  Witness $w_m$ can't be in region $A$ either as
  it would remove both edges $v_{b1} v_{m2}$ and $v_{b2} v_{m1}$.
  Indeed, $\measuredangle v_{m1} w_m v_{m2}$ is smaller than
  $\measuredangle v_{b1} w_m v_{m2}$ and $\measuredangle v_{m1} w_m
  v_{b2}$.  Therefore $w_m$ must be in region $B \subset \triangle
  v_{b1} w_b v_{b2}$, implying $\measuredangle v_{b1} w_m v_{b2} >
  \measuredangle v_{b1} w_b v_{b2} \geq 90^\circ$.
  But then $w_m$ eliminates both $v_{b1} v_{m2}$ and $v_{b2} v_{m1}$,
  a contradiction.
\end{proof}

\begin{lemma}
  \label{bi-chromatic triangulation}
  Consider a set of points, colored by two or more colors, in convex
  position, such that there are no two consecutive black points.
  There is a triangulation of this set of points such that every
  triangle has exactly one black vertex.
\end{lemma}
\begin{proof}
Find three consecutive points $a$, $b$, $c$, on the convex hull of the set of points such that exactly one of them is black.
Add the triangle $\triangle abc$ to the triangulation and remove $b$ from the set of points.
Repeat this procedure until no three consecutive points $a$, $b$, $c$ with exactly one black are found.
At this moment, there must be either only two points left (in which case we have constructed the desired triangulation) or more than two points but no black.

Indeed, notice that if there were more than two points left, with at least one black one, this black point would have as neighbors points of a different color, and we could repeat the procedure described above at least once more.
This is true as during the procedure above in which we add a triangle to the triangulation and remove one point, two cases may occur.
In the first case, one black point is removed, in which case the neighborhood of all other black points does not change (as the black point removed had colored points as neighbors), and the original condition that there are no two consecutive black points is maintained.
In the other case, one point $b$ of a different color than black is removed and the triangle $abc$ is added to the triangulation.
Either $a$ or $c$ might be black, but not both as all triangles are
incident to exactly one black point. Suppose without loss of
generality that $a$ is black. Then $a$ gets as a neighbor, instead of
$b$, a new point $c$ of a different color than black, and once again
the original condition of no two consecutive black points is maintained.

If there are more than two points left, none black, remove the last triangle added and put back the point $b$ that was removed last; $b$ must be black.
If $b$, $c_1$, $c_2$, $\ldots$, $c_{m}$ are the remaining points, in order, add the triangles $\triangle b c_1 c_2$, $\triangle b c_2 c_3$, $\ldots$, $\triangle b c_{m-1} c_{m}$.
The set of points is triangulated such that each triangle contains exactly one black vertex.
\end{proof}

\begin{lemma}
\label{max subset}
Consider the convex hull $H$ of the vertices of a witness Gabriel
drawing of a complete $k$-partite graph.  Consider a subset of
vertices of $H$, with at least one black point and no two consecutive
black points.  The interior of the convex hull of this subset is empty
of black points.
\end{lemma}
\begin{proof}
By Lemma~\ref{bi-chromatic triangulation}, this subset of points can be triangulated such that each triangle has exactly one black vertex.
In every triangle, the two edges incident to the black vertex are in the complete $k$-partite witness Gabriel drawing.
Any black vertex $b_2$ inside the convex hull would be inside a triangle defined by two edges incident to a black vertex $b_1$.
By Lemma~\ref{wedge}, $b_2$ would be incident to $b_1$, a contradiction.
\end{proof}

Consider a finite set of points $P$ colored with $k$ colors.  A
(\emph{quasi-convex circular}) \emph{quasi-ordering} $\quasi$ of $P$
is a partition of $P$ into $s\geq k$ subsets $P_1,\ldots,P_s$,
cyclically ordered as $P_1 \quasi P_2 \quasi \ldots \quasi P_s \quasi
P_1$, such that (a) every $P_i$ contains only points of one color and
consecutive sets have different colors, and (b) any subset $S \subset
P$ with at most two elements from each $P_i$ is in convex position and
their cyclic ordering along $\CH(S)$ is consistent with $\quasi$; we
make no assumption on the relative order of the points coming from the
same set $P_i$ (hence the choice of the term \emph{quasi}-convex
\emph{quasi}-ordering).  Refer to Figure~\ref{fig:disaster}(left).
\begin{figure}
  \centering
  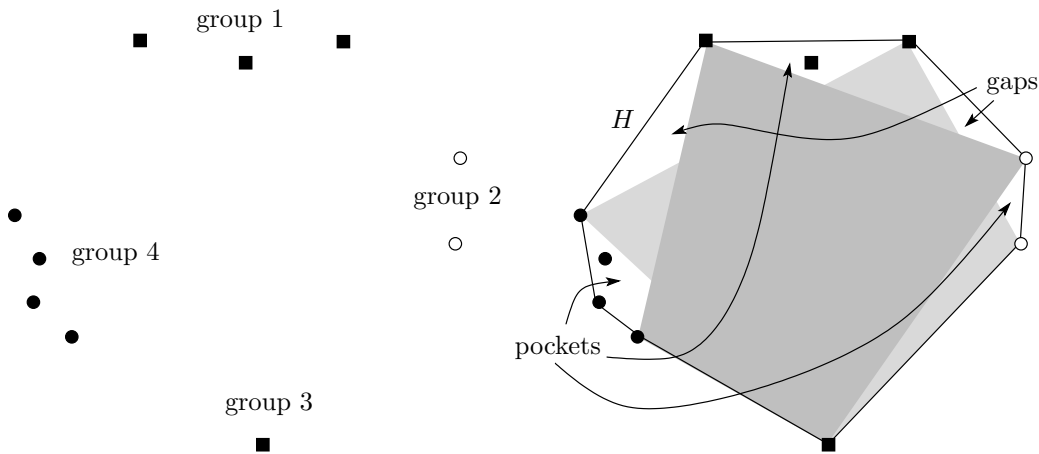
  \caption{Left:~An example of a quasi-ordered set.
    Right:~Quasi-ordering of the vertex set of a witness Gabriel
    drawing of a complete $k$-partite graph.  Sets $K_{\textsc{cw}}$
    and $K_{\textsc{ccw}}$ are shaded.  Pockets and gaps form the
    complement of $K_{\textsc{cw}} \cup K_{\textsc{ccw}}$ in $H$.}
  \label{fig:disaster}
\end{figure}

\begin{lemma}
\label{order}
The set of vertices $P$ of a witness Gabriel drawing of a complete
$k$-partite graph can be quasi-ordered.
\end{lemma}
\begin{proof}
  Consider the convex hull $H\mathop{:=}\CH(P)$ of the vertices. 
  Let $E \subset P$ contain the vertices that appear on $H$.
  Traversing $H$ in counterclockwise order, group consecutive vertices
  of the same color together, forming a cyclically ordered partition
  $E_1, \ldots, E_s$ of $E$.  Since $E$ is in convex position, this
  partition is trivially a quasi-ordering, as define above.  We now
  show how to extend it to the entire set $P$.  Specifically, for each
  point in $P'\mathop{:=}P \setminus E$, we assign it to one of the
  groups $E_i$.
  We say a group $E_i$ is \emph{trivial} if it contain just one
  point.  We will see that no other point is ever added to such a
  group.  

  Let $K_{\textsc{cw}}$ be the convex hull of the set consisting of
  the most \emph{clockwise} point from each $E_i$.  Analogously define
  $K_{\textsc{ccw}}$ for most counterclockwise points.

  Lemma~\ref{max subset} implies that no point of $P'$ is contained
  inside $K_{\textsc{cw}} \cup K_{\textsc{ccw}}$.  The remaining
  points of $P$ lie in $H \setminus (K_{\textsc{cw}} \cup
  K_{\textsc{ccw}})$.  This set naturally splits into at most $2s$
  convex subsets: \emph{pockets} $\Pi_i$ are connected components of
  this set whose boundary contains one (and only one) non-trivial
  group $E_i$, while \emph{gaps} are components adjacent to two
  consecutive groups (refer to Figure~\ref{fig:disaster}(right)).

  We complete our quasi-ordering now: We claim that $P' \subset
  \bigcup \Pi_i$, i.e., there are no points of $P'$ in the gaps.  And
  a point in pocket $\Pi_i$ is simply assigned to $E_i$.

  It remains to prove (i)~the emptiness of gaps and (ii)~the fact that
  the resulting partition of $P$ is a quasi-ordering.  We start
  with~(i), for which it is sufficient to argue that a gap between
  consecutive groups, say, white $E_1$ and black $E_2$ is empty of
  points of $P'$.  Let $w$ and $b$ be the points of $E_1$ and $E_2$,
  respectively, adjacent to the gap.  For a contradiction, consider a
  point $p$ in the gap.  Suppose first that it is not black or white,
  say green.  By Lemma~\ref{2 vertices} a green point $g$ appears on
  $H$.  By construction, the gap lies in $\triangle bwg$ and therefore
  so does $p$, contradicting Lemma~\ref{wedge}, as it forces the
  existence of a green-green edge $pg$.  Thus $p$ must be white
  or black.  We assume it is black, without loss of generality.  Then
  we again take a green point $g$ on $H$, forcing $\triangle bwg$ to
  contain $p$ and ensuring the existence of a black-black edge $bp$,
  by Lemma~\ref{wedge}---a contradiction.  Therefore, indeed, the
  gaps are empty.

  We now argue (ii): the resulting partition of points is indeed a
  quasi-ordering.  We start by proving part (a) of the definition,
  namely that all points in the pocket $\Pi_i$ have the color of $E_i$
  (consecutive groups have different colors by construction).  For a
  contradiction, suppose a, say, blue point $b \in \Pi_i \cap P'$ lies
  in the pocket of purple group $E_i$.  Let $b'$ be a blue point on
  $H$.  By definition of a pocket, $b$ lies in a triangle formed by
  $b'$ and two purple points of $E_i$, once again contradicting
  Lemma~\ref{wedge} and hence~(a) is proved.

  It remains to check that any subset of $P$ formed by taking at most
  two points from each pocket $\Pi_i$ (including $E_i$) is in convex
  position.  If at most one point is used from each group, the
  assertion holds by construction.  To finish the argument, it is
  sufficient to show that, for any two, say, cyan points $c$ and $c'$,
  the line $cc'$ leaves all the other pockets on the same side; indeed
  it is sufficient to prove this for the two pockets adjacent to the
  cyan pocket of $c$ and $c'$ and thus of colors other than cyan.  If
  the line $cc'$ did not have both pockets entirely to one side of it,
  there would be two points $p$ and $q$ coming from these two pockets
  (one from each, or both from the same one) of color other than cyan,
  on opposite sides of the line $cc'$.  Since $c$ and $c'$ lie in the
  same pocket and therefore on the same side of $pq$, this would force
  either $c \in \triangle pqc'$ or $c' \in \triangle pqc$,
  contradicting once again Lemma~\ref{wedge} and thereby completing
  the proof of the Lemma.
\end{proof}

\begin{lemma}
  \label{lem:K2222-convex}
  Any witness Gabriel drawing of $K_{2,2,2,2}$ must have vertices in
  convex position.
\end{lemma}

\begin{proof}
By Lemma~\ref{2 vertices}, there are at least 2 vertices of each color on the convex hull.
\end{proof}

\begin{lemma}
\label{K2222 w}
In a witness Gabriel drawing of a $K_{2,2,2,2}$, there is no witness inside the convex hull of the set of vertices.
\end{lemma}
\begin{proof}
By Lemma~\ref{lem:K2222-convex}, all the vertices are in convex position.
Take any triangulation of the set of vertices.
Each triangle will be incident to at most two vertices of the same color; therefore, for each triangle, at least two edges will be present in the witness Gabriel drawing.
By Lemma~\ref{wedge}, there can't be any witness in any of the triangles.
\end{proof}

\begin{lemma}
  \label{$K_{2,2,2,2}$}
  There is no witness Gabriel drawing of $K_{2,2,2,2}$ in which all
  the vertices of the same color are consecutive (see
  Figure~\ref{GG4PartiteN8}(left)).
  \begin{figure}
    \centering
    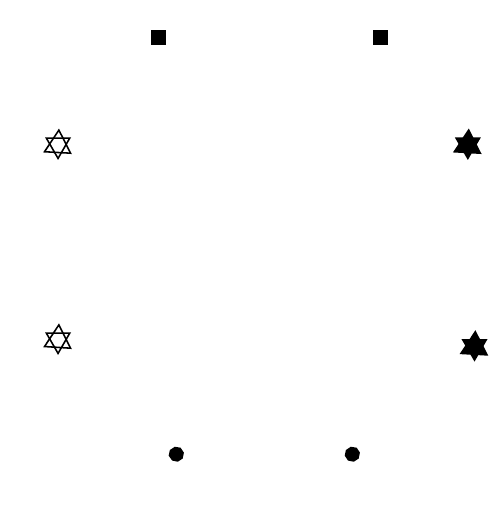%
    \hspace{0pt plus 1fil}%
    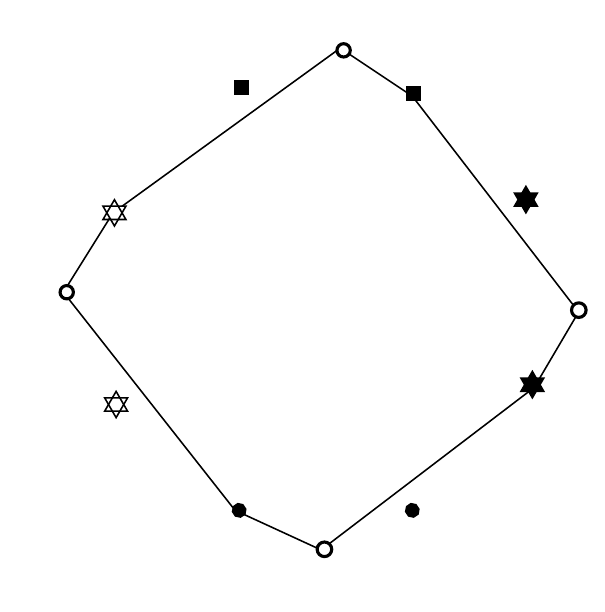
    \caption{Left: A drawing of $K_{2,2,2,2}$ in which all the
      vertices of the same kind are consecutive.  Right:~Octagon in a
      4-partite witness Gabriel drawing.}
    \label{GG4PartiteN8}
  \end{figure}
\end{lemma}

\begin{proof}
  By Lemma~\ref{lem:K2222-convex}, all the vertices are in convex
  position, so the ordering of vertices is well defined.  Name the
  vertices $v_1, \ldots, v_{8}$ clockwise as in
  Figure~\ref{GG4PartiteN8}(right).  The witnesses $w_1,\ldots, w_4$
  eliminate the edges $v_8 v_1$, $v_2 v_3$, $v_4 v_5$, $v_6 v_7$,
  respectively.  By Lemma~\ref{K2222 w}, the witnesses are outside the
  convex hull of all the vertices.  As we will see below, the four
  witnesses are distinct and necessary.  For a contradiction, suppose one
  witness $w=w_2=w_3$ removes two monochromatic edges, say $v_2 v_3$ and
  $v_4 v_5$.  The witness $w$ sees $v_2 v_3$ and $v_4 v_5$,
  respectively, with an angle larger than $90^\circ$,
  i.e. $\measuredangle v_2 w v_3 >90^\circ$ and $\measuredangle v_4 w
  v_5 >90^\circ$.  As we already saw, $w$ is outside $\CH(P)$, and
  therefore sees it with a maximum view angle smaller than
  $180^\circ$.  Hence two cases are possible:
  \begin{enumerate}
  \item $w$ sees the two edges overlapping, and without loss of generality, it sees the vertices from left to right in the following order: $v_2$, $v_4$, $v_3$, $v_5$.
  But then $w$ removes $v_2v_5$, since $\measuredangle v_2 w v_5 >
  \measuredangle v_2w v_3$. 
  \item $w$ sees the two edges nested, and without loss of generality, it sees the vertices from left to right in that order: $v_2$, $v_4$, $v_5$, $v_3$.
  But then $w$ removes $v_2v_5$ and $v_3v_4$, by similar reasoning.
  \end{enumerate}
Hence we can conclude that each witness $w$ removes exactly one
monochromatic edge, and four distinct witnesses are necessary.


Consider the octagon $w_1 v_1 w_2 v_3 w_3 v_5 w_4 v_7$ (see Figure~\ref{GG4PartiteN8}(right)).
The interior angles at $w_i$ measure less than $90^\circ$ each; otherwise a witness would be inside a diametral disk of two vertices of different colors.

The interior angles at $v_1$, $v_3$, $v_5$, $v_7$ measure strictly
less than $180^\circ$ each.
Indeed if one of these angles were equal to $180^\circ$, we would have three points on a line; this contradicts our assumption of general position.
Now suppose that one of these angles, say, $\measuredangle w_1 v_7 w_4
> 180^\circ$.
By definition of the witness Gabriel drawing, we have $\measuredangle v_1 w_1 v_7 < 90^\circ$ and $\measuredangle v_1 w_1 v_8 \geq 90^\circ$ (see Figure~\ref{GG4PartiteProof}).
\begin{figure}
  \centering
  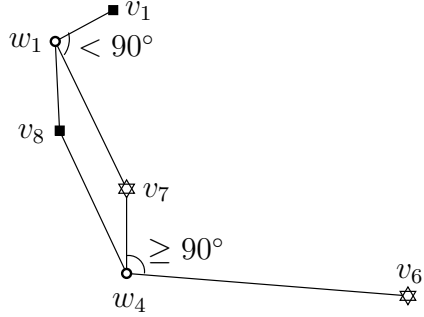
  \caption{Detail of the proof of Theorem \ref{$K_{2,2,2,2}$}.}
  \label{GG4PartiteProof}
\end{figure}
As $\measuredangle v_7 w_4 v_6 \geq 90^\circ$, we would have $\measuredangle v_8 w_4 v_6 > 90^\circ$, a contradiction.

Therefore the sum of the interior angles of this octagon is less that $1080^\circ$, which is impossible.
\end{proof}

The constraints described in the preceding results lead to a graph
that is not drawable:
\begin{theorem}
  There is no witness Gabriel drawing of $K_{3,3,3,3}$.
\end{theorem}

\begin{proof}
 
 
 Assume such a drawing exists. We consider the ordering of the colors of vertices of $K_{3,3,3,3}$, in the sense of Lemma~\ref{order}. In the case analysis below, we argue that the only ordering of the vertices consistent with Lemmas~\ref{K2222 w} and \ref{$K_{2,2,2,2}$} is such that all the vertices are in convex position and between every pair of consecutive vertices of one color, there is exactly one vertex of every other color (see Figure~\ref{GGK3333bb}). 
\begin{figure}
  \centering
  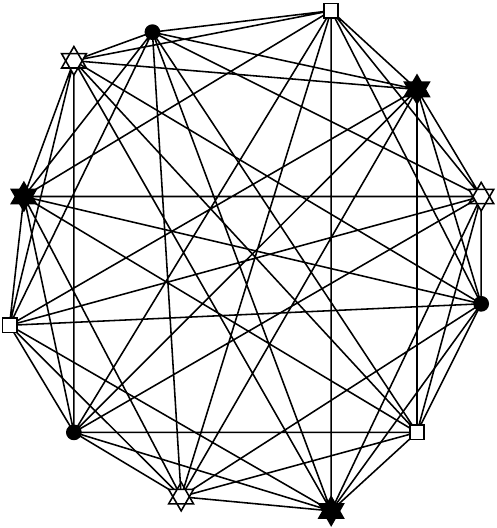
  \caption{A tentative witness Gabriel drawing of $K_{3,3,3,3}$.}
  \label{GGK3333bb}
\end{figure}

 

All the possible ways to order three points of two different colors
using the ordering defined in Lemma~\ref{order} are in
Figure~\ref{GGK3333CaseAnalysis1}; notice in two of the three cases,
the points must be in convex position by Lemma~\ref{order}, and in the
remaining (middle) case we must have the colors separated by a line and situated so that any choice of two points of each color is in convex position.
\begin{figure}
  \centering
  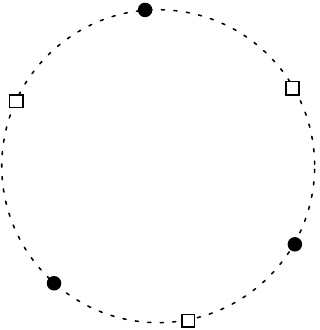%
  \hspace{1cm}%
  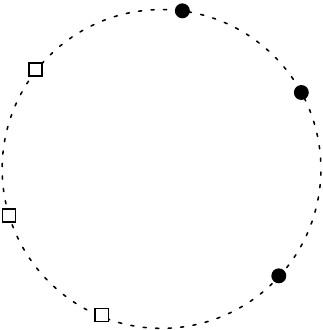%
  \hspace{1cm}%
  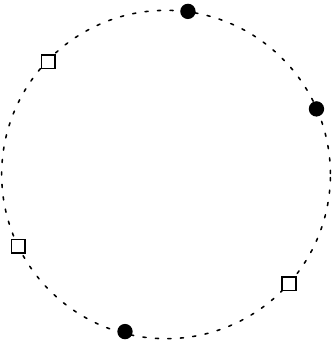
  \caption{All possible ways to order the vertices of a witness
    Gabriel drawing of $K_{3,3}$.}
  \label{GGK3333CaseAnalysis1}
\end{figure}
All the ways to add three points of a third color to the cases in
Figure~\ref{GGK3333CaseAnalysis1} without violating Lemma~\ref{GGK222}
are in Figure~\ref{GGK3333CaseAnalysis11}.  We draw the points on a
circle for ease of visualization.  Again, they must be in convex
position unless there is a group of three consecutive points of the
same color (second and third figures in the top row), in which case
these points need not all appear on the convex hull of the entire set;
see Lemma~\ref{order}. 
\begin{figure}
  \centering
  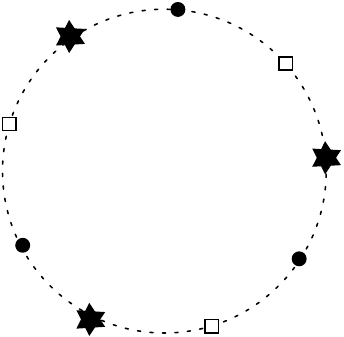
  \hspace{1cm}
  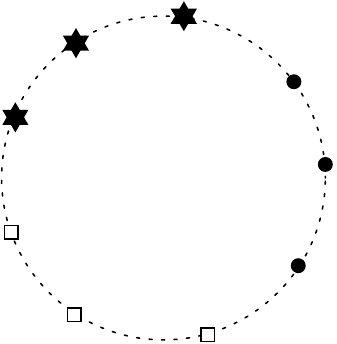
  \hspace{1cm}
  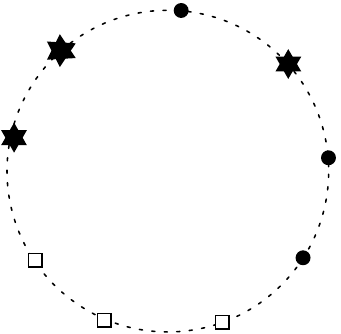\\[1cm]
  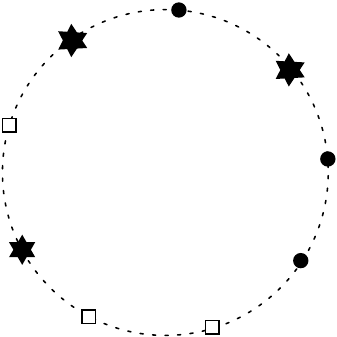
  \hspace{1cm}
  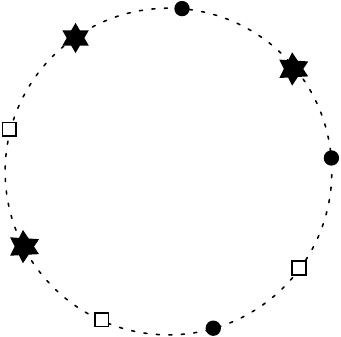
  \caption{All possible ways to order the vertices of a witness
    Gabriel drawing of $K_{3,3,3}$.}
  \label{GGK3333CaseAnalysis11}
\end{figure}

There is only one way to add three points of a fourth color to the set of points of Figure~\ref{GGK3333CaseAnalysis11} without violating Lemma~\ref{GGK222} and Lemma~\ref{$K_{2,2,2,2}$} (see Figure~\ref{GGK3333CaseAnalysis111}).
\begin{figure}
  \centering
  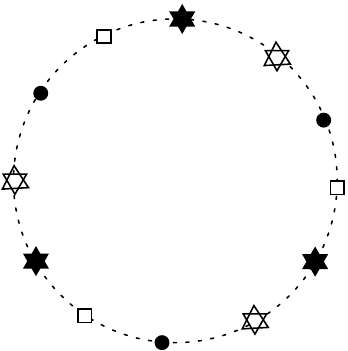
  \caption{The only way to order the vertices of a witness Gabriel
    drawing of a $K_{3,3,3,3}$.}
  \label{GGK3333CaseAnalysis111}
\end{figure}
Notice that in this Gabriel drawing of  $K_{3,3,3,3}$, by Lemma~\ref{2 vertices}, all the points are in convex position.

Now we will show that the tentative witness Gabriel drawing of
$K_{3,3,3,3}$, depicted in Figure~\ref{GGK3333bb}, where vertices are
in convex position and such that between every pair of consecutive vertices
of one color there is exactly one vertex of each other color, cannot
be realized.

Consider the hexagon formed by the three black-star vertices $b_1$, $b_2$, and $b_3$, and the three witnesses $w_1$, $w_2$, and $w_3$  that remove the edges $b_1 b_2$, $b_2 b_3$, and $b_3 b_1$, respectively (see Figure~\ref{GGK3333c}).
\begin{figure}
  \centering
  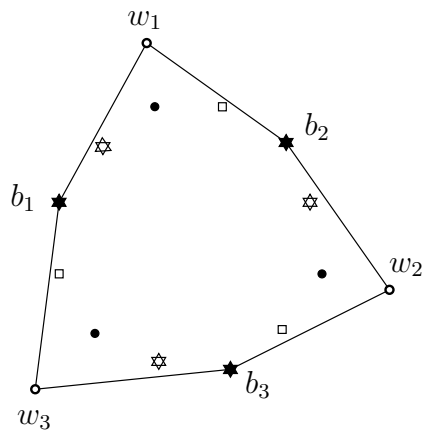
  \caption{tentative Witness Gabriel drawing of $K_{3,3,3,3}$.}
  \label{GGK3333c}
\end{figure}
The three witnesses are distinct as otherwise they would remove some
bichromatic edges.  The measure of each of the three interior angles
$\angle b_1 w_1 b_2$, $\angle b_2 w_2 b_3$, and $\angle b_3 w_3 b_1$
are at least $90^\circ$.  The sum of the measures of interior angles
in a hexagon is $ 720^\circ$.  Therefore, $\measuredangle w_3 b_1
w_1$, $\measuredangle w_1 b_2 w_2$, and $\measuredangle w_2 b_3 w_3$
sum up to at most $450^\circ$.

If one repeats the argument for each of the four colors and their
corresponding witnesses, one obtains that the sum of the interior
angles such that the vertex of the angle is a vertex of the graph
adjacent to two witnesses, is at most $1800^\circ$.  However, the sum
of the interior angles of a $12$-gon that is the convex hull of the
vertices equals $1800^\circ$.  Therefore for at least one color, say
black-star without loss of generality, and because of our assumption
of general position, at least one point of another color will be
outside of the hexagon $ b_1 w_1 b_2 w_2 b_3 w_3$, and a bi-chromatic
edge will be eliminated.  
\end{proof}

From the preceding result we immediately obtain the following:
\begin{cor}
  No graph containing $K_{3,3,3,3}$ as an induced subgraph can be drawn as a witness Gabriel
  graph. In particular, there is no witness Gabriel drawing of
  $K_{p,q,r,s}$ for $p,q,r,s \geq 3$.
\end{cor}

\section{Construction Algorithms}

In this section we describe two algorithms to compute the witness
Gabriel graph $\GG^{-}(P,W)$ from two given sets of points $P$ and $W$.

\begin{theorem}
  Given two point sets $P,W$ with $\Pabs{P}+\Pabs{W} =n$, the graph $\GG^{-}(P,W)$
  can be computed in $\Theta (n^2)$ time.
 
\end{theorem}

 It is clear that in the worst case $\Omega (n^2)$  time is required, since the graph may have $\Theta(n^2)$ edges.

\paragraph*{First algorithm:} For each point $p \in P$, do
the following: For each point $q \in W$, draw the line $l_q$ through
$q$, perpendicular to $pq$.  Consider the interior of the intersection
$I_p$ of the half-planes containing $p$ bounded by the lines
$l_q$, $\forall q \in W$.  Then, an edge $pr$, $r \in P \setminus
\{p\}$, is in $\GG^{-}(P,W)$ if and only if $r \in I_p$ (see Figure~\ref{buildGGG}).
Indeed any point $r \in P \setminus \{p\}$ in the interior of $I_p$ will make an angle $\angle r q p < 90^\circ$ with any $q \in W$.
On the other hand, any point $r \in P \setminus \{p\}$ on the boundary or outside $I_p$, will make an angle $\angle r q p \geq 90^\circ$ for at least one $q \in W$.

Once we have computed the circular ordering of points in $P \cup W$
around $p$, we can compute $I_p$ and identify all edges $pr$ in linear
time, for a fixed $p$.  The circular ordering, for all $p$, can be
computed in quadratic time by standard methods using the dual
arrangement of $P \cup W$ \cite{EOS86}.

\begin{figure}
  \centering
  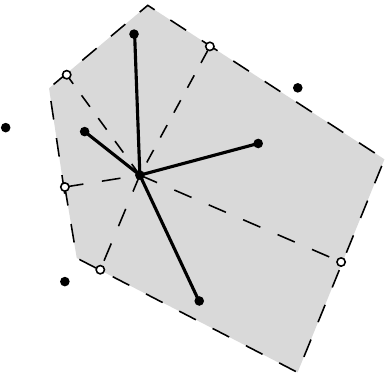
  \caption{The first algorithm to build $\GG^{-}(P,W)$. Black points
    are in $ P$ and white points in $ W$.}
  \label{buildGGG}
\end{figure}

\paragraph*{Second algorithm:} Build the Voronoi diagram
$\Vor(W)$ of the points in $W$.  For each $p \in P$, add the point
$p$ to $\Vor(W)$ and consider all segments of the form $pr$ with $r
\in P \setminus \{p\} $.  For each edge $pr$ take the midpoint
$m(p,r)$.  Observe that $m(p,r)$ is in the Voronoi cell of $p$ in
$\Vor(W\cup \{p\})$ if and only if the edge $pr$ is in $\GG^{-}(P,W)$.
The algorithm can be implemented to run in quadratic time using
standard tools.  Again, it is useful to have the circular ordering
of all points in $P \cup W$ around each point in $P$.

As a final observation, it is worth mentioning an algorithm that would
be more efficient in some cases, but not in the worst case.  
The \emph{witness Delaunay graph} of a point set $P$
in the plane, with respect to point set $W$ of witnesses, denoted
$\DG^{-}(P,W)$, is the graph with vertex set $P$ in which two points
$x,y\in P$ are adjacent when there is a disk whose boundary passes
through $x$ and $y$ and whose interior does not contain any witness
$q\in W$.  This graph was 
introduced in \cite{ADH}, and an algorithm for its computation with
running time $O(e \log n + n \log^2 n)$, where $e$ is the number of
edges in the graph, was also described there.

Now, as $\GG^{-}(P,W)$ is a subgraph of $\DG^{-}(P,W)$, once the latter
graph has been computed we can easily check in $O(\log n)$ time
whether one of its edges, say $pq$ belongs to $\GG^{-}(P,W)$: if $m$ is the
midpoint of $pq$, we only have to find the point $z$ from $W$ which is
closest to $m$, which can be achieved by standard point location in
$\Vor(W)$.  Once $z$ has been obtained, $pq\in \GG^{-}(P,W)$ if and only if
$d(m,z)>d(m,p)$. Therefore $\GG^{-}(P,W)$ can be computed in additional $O(e
\log n)$ time, once $\DG^{-}(P,W)$ is available.  To summarize, we can
compute $\GG^{-}(P,W)$ in time $O(e \log n + n \log^2 n)$, where $e\geq k$ is the number of edges in $\DG^{-}(P,W)$ and $k$ is the number of edges in $\GG^{-}(P,W)$.

\section{Verification Algorithm}

In this section we present an algorithm to verify whether a graph
$G=(V,E)$ embedded in the plane can be a witness Gabriel graph
$\GG^{-}(V,W)$, for some suitable set of witnesses $W$.

\begin{theorem}
  \label{thm:verification}
  Given a straight-line graph $G=(V,E)$ embedded in the plane,
  checking if there exists a set of witnesses $W$ so that $G$
  coincides with $\GG^{-}(V,W)$ can be done in $O(\Pabs{V}^2 \log \Pabs{E})$ time; if
  the answer is positive, such a set of witnesses $W$ can be computed
  within the same time bounds.
\end{theorem}

\noindent
\textbf{Algorithm:} For each edge $pq$ in $G$, draw a disk $D_{pq}$
with diameter $pq$.  Take the union $U = \bigcup_{pq \in E(G)} D_{pq}$
of these disks.  Compute the Voronoi Diagram of the arcs and vertices of
the boundary of $U$ \cite{Yap}.  For each pair of vertices $r$ and $s$ such that
there is no edge between them in $G$, draw a disk $D_{rs}$ with
diameter $rs$.  Check if the center $c$ of $D_{rs}$ lies in $U$.  
If it does not, $c$ (or any point sufficiently close to it) is a valid witness for $rs$.
 If it does, find which cell $C$ of the Voronoi diagram contains $c$ and
check if the site of $C$ intersects $D_{rs}$. If the site of $C$ does not intersect $D_{rs}$,
 $D_{rs} \subset U$, and it is impossible to place a witness to eliminate $rs$
without also eliminating a legitimate edge of $G$.  Therefore $G$ is
not a witness Gabriel graph $\GG^{-}(V,W)$, for any $W$ (see Figure~\ref{checkGGG}).
\begin{figure}
  \centering
  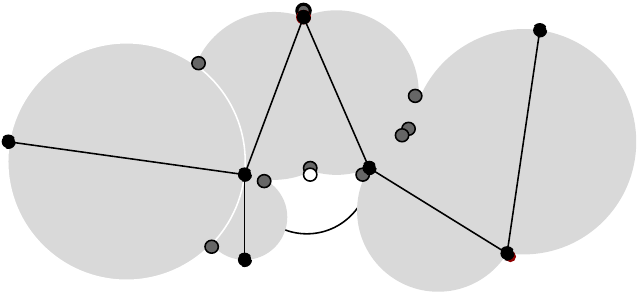
  \caption{The algorithm to check if a geometric graph is a witness
    Gabriel drawing.  Black points are in $P$.}
  \label{checkGGG}
\end{figure}
Otherwise, a suitable witness in $D_{rs}\setminus U$ is easy to identify.
We continue to the next non-edge $rs$.

If none of the tests fail, we have produced a set $W$ of witnesses
such that $G=\GG^{-}(P,W)$.

The algorithm can be implemented to run in time $O(\Pabs{V}^2 \log \Pabs{E})$
using standard tools.

\section{Final Remarks}

We have described in this paper several properties of the witness
Gabriel graph, as well as algorithms for its computation and
verification.  However, we have omitted the description of some
extensions.  For example, as the standard Gabriel graph can be
extended to higher order, this can be done for the witness
generalization: In a \emph{witness $k$-Gabriel graph}, an edge $ab$,
$a,b \in P$, is in the graph if there are fewer than $k$ witnesses in
$D_{ab}\setminus \{a,b\}$.  Most of the preceding results can be
easily modified to provide the corresponding conclusions about witness
$k$-Gabriel graphs.

There are some obvious open problems left in this paper, such as
closing the gaps between some bounds. In particular, it would be
interesting to tighten the bounds in Theorem~\ref{thm:eliminate-all}
on the maximum number of witnesses needed to eliminate all edges in a
witness Gabriel graph.  Perhaps more embarrassingly, we have no linear
(nor, in fact \emph{any subquadratic}) upper bound on the number of
witnesses that are sufficient to realize an arbitrary witness Gabriel
graph (Theorem~\ref{thm:eliminate-some}).

On the algorithmic side, designing an \emph{output-sensitive}
algorithm for constructing a witness Gabriel graph, given its set of
vertices and witnesses, i.e., one whose running time depends on the
number of edges in the graph, is still an open problem.  An ideal
algorithm would pay a small, say, polylogarithmic, cost per additional
graph edge.  Another issue is whether finding the \emph{minimum}
number of witnesses required to realize a given geometric graph in the
plane as a witness Gabriel graph (as in
Theorem~\ref{thm:verification}) is NP-hard, or whether it can be solved
in polynomial time.

Finally, we also mention that some natural long-term goals, such as a
complete characterization of the class of witness Gabriel graphs or
the design of efficient algorithms testing graphs for membership,
remain elusive to date, which, on the other hand, is a common
situation for most classes of standard proximity graphs.

\bibliography{bqual}
\end{document}